\documentclass[12pt]{article}

\usepackage[pdfencoding=auto,bookmarks=true,bookmarksnumbered=true,colorlinks=true]{hyperref}
\usepackage{float}
\usepackage{amsthm,amssymb,color, epsfig, graphics}
\usepackage{xcolor}
\usepackage{graphicx}
\usepackage[intlimits]{amsmath}
\usepackage{latexsym}
\usepackage{txfonts}
\usepackage{url}
\newtheorem{proposition}{Proposition}
\newtheorem*{acknowledgments}{Acknowledgments}

\theoremstyle{definition}
\newtheorem{definition}{Definition}
\newtheorem{ex}{Example} 

\newtheorem{conjecture}{Conjecture}

\newtheorem{remark}{Remark}

\newcommand{\R}{\mathbb{R}}

\newcommand{\Z}{\mathbb{Z}}

\newcommand{\Tt}{\mathcal{T}}
\newcommand{\Tw}{T\!w}
\renewcommand{\Wr}{W\!r}

\newcommand{\doubleu}{w}

\usepackage[center,cam]{crop}

\makeatletter
\@addtoreset{equation}{section} 

\makeatother

\begin{document}

%

%


%
\begin{center}
\begin{Large}
Linkage Mechanisms Governed by Integrable Deformations of Discrete Space Curves\\[5mm]
\end{Large}
\begin{normalsize}
Shizuo {\sc Kaji}$^1$, Kenji {\sc Kajiwara}$^2$\\[1mm]
Institute of Mathematics for Industry, Kyushu University\\
744 Motooka, Fukuoka 819-0395, Japan\\
e-mail: {}$^1$skaji@imi.kyushu-u.ac.jp,\ 
$^2$kaji@imi.kyushu-u.ac.jp\\[3mm]
Hyeongki {\sc Park}\\[1mm]
Graduate School of Mathematics, Kyushu University\\
744 Motooka, Fukuoka 819-0395, Japan\\
e-mail: h-park@math.kyushu-u.ac.jp\\[2mm]
\end{normalsize}
\end{center}
\begin{abstract}
  \noindent
  A \emph{linkage mechanism} consists of rigid bodies assembled by joints which can be used to translate and transfer motion
  from one form in one place to another.
  In this paper, we are particularly interested in a family of spatial linkage mechanisms
  which consist of $n$-copies of a rigid body joined together by hinges to form a ring.
Each hinge joint has its own axis of revolution and rigid bodies joined to it can be freely rotated around the axis.
The family includes the famous threefold symmetric Bricard 6R linkage, also known as the Kaleidocycle,
 which exhibits a characteristic ``turning-over'' motion.
We can model such a linkage as a discrete closed curve in $\mathbb{R}^3$ of constant torsion up to sign.
Then, its motion is described as the deformation of the curve preserving torsion and arc length.
We describe certain motions of this object that are governed by the semi-discrete mKdV and sine-Gordon equations,
where infinitesimally the motion of each vertex is confined in the osculating plane.
\end{abstract}


\section{Introduction}
A linkage is a mechanical system consisting of rigid bodies (called \emph{links}) joined together by
\emph{joints}.  They are used to transform one motion to another as in the famous Watt parallel
motion and a lot of examples are found in engineering as well as in natural creatures (see, for
example, \cite{Chen2011}).

Mathematical study of linkage dates back to Euler, Chebyshev, Sylvester, Kempe, and Cayley and since
then the topology and the geometry of the configuration space have attracted many researchers (see
\cite{Faber2008,Kapovich-Millson2002,Magalhaes} for a survey).  Most of the research focuses on
\emph{pin joint linkages}, which consist of only one type of joint called pin joints.  A pin joint
constrains the positions of ends of adjacent links to stay together.  To a pin joint linkage we can
associate a graph whose vertices are joints and edges are links, where edges are assigned its
length.  The state of a pin joint linkage is effectively specified by the coordinates of the joint
positions, where the distance of two joints connected by a link is constrained to its length.  Thus,
its configuration space can be modelled by the space of isometric imbeddings of the corresponding
graph to some Euclidean space.  Note that in practice, joints and links have sizes and they collide
to have limited mobility, but here we consider ideal linkages with which joints and links can pass
through each other.

While the configuration spaces of (especially planar) pin joint linkages are well studied, there are
other types of linkages which are not so popular.  
In this chapter, we are mainly interested in
linkages consisting of hinges (revolute joints).  To set up a framework to study linkages with
various types of joints, we first introduce a mathematical model of general linkages as graphs
decorated with groups (\S \ref{sec:linkage-graph}), extending previous approaches (see
\cite{Muller2015} and references therein).  This formulation can be viewed as a special type of
constraint network (e.g., \cite{Freuder}).  Then in \S \ref{sec:hinged-linkage}, we focus on
linkages consisting of hinges.  Unlike a pin joint which constrains only the relative positions of
connected links, a hinge has an axis so that it also constrains the relative orientation of
connected links.

We are particularly interested in a simple case when $n$ links in $\R^3$ are joined by hinges to
form a circle (\S \ref{circle-system}).  Such a linkage can be roughly thought of as a discrete
closed space curve, where hinge axes are identified with the lines spanned by the binormal vectors.
Properties of such linkages can thus be translated and stated in the language of discrete curves.  An
example of such linkage is the threefold symmetric Bricard 6R linkage consisting of six hinges
(Fig. \ref{fig:bricard}), which exhibits a turning-over motion and has the configuration space
homeomorphic to a circle.  As a generalisation to the threefold symmetric Bricard 6R linkage, we
consider a family of linkages consisting of copies of an identical links connected by hinges, which
we call \emph{Kaleidocycles}, and they are characterised as discrete curves of constant speed and
constant torsion.

\begin{figure}[ht]
\center
\includegraphics[width=5cm]{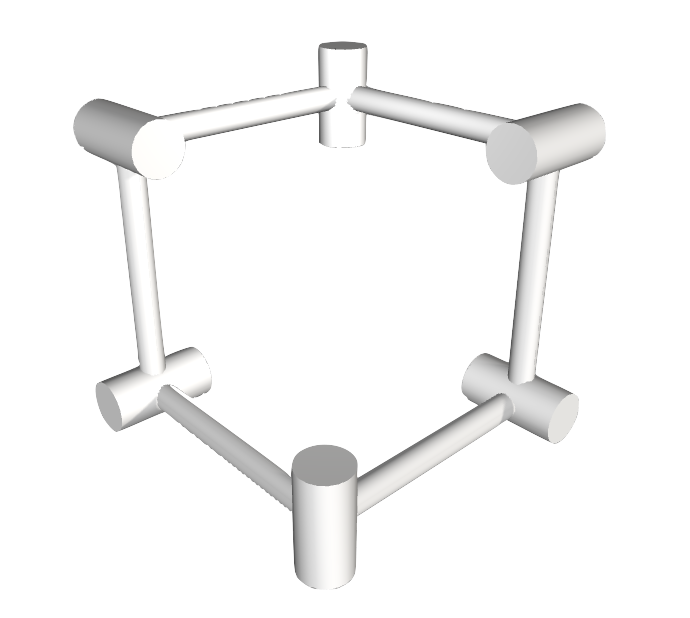}
\caption{Threefold symmetric Bricard 6R linkage.}
\label{fig:bricard}
\end{figure}

The theory of discrete space curves has been studied by many authors.  The simplest way to
discretise a space curve is by a polygon, that is, an ordered sequence of points
$\gamma=(\gamma_0,\gamma_1,\ldots), \gamma_i \in \R^3$.  Deformation of a curve is a
time-parametrised sequence of curves $\gamma(t)$, where $t$ runs through (an interval of) the real
numbers. Deformations of a given smooth/discrete space curve can be described by introducing an
appropriate frame such as the Frenet frame, which satisfies the system of linear partial
differential/differential-difference equations. The compatibility condition gives rise to nonlinear
partial differential/differential-difference equation(s), which are often integrable.  It is
sometimes possible to construct deformations for the space curves using integrable systems which
preserve some geometric properties of the space curve such as length, curvature, and torsion. For
example, a deformation is said to be \emph{isoperimetric} if the deformation preserves the arc
length.  In this case, the modified Korteweg-de Vries (mKdV) or the nonlinear Schr\"odinger (NLS) equation
and their hierarchies naturally arise as the compatibility condition
\cite{Calini-Ivey1999,Doliwa-Santini:PLA,Hasimoto,inoguchi2014discrete,Lamb,Langer-Perline,Nakayama_Segur_Wadati:PRL,Rogers-Schief:book}.
Various continuous deformations for the discrete space curves have been studied in
\cite{Doliwa-Santini:JMP,Hisakado-Wadati,inoguchi2014discrete,Nakayama:JPSJ2007,Nishinari}, where
the deformations are described by the differential-difference analogue of the mKdV and the NLS
equations.

The motion of Kaleidocycles corresponds to isoperimetric and torsion-preserving deformation of
discrete closed space curves of constant torsion.  In \S \ref{sec:motion}, we define a flow on the
configuration space of a Kaleidocycle by the differential-difference analogue of the mKdV and the sine-Gordon equations
(\textit{semi-discrete mKdV and sine-Gordon equations}).  This flow generates the characteristic turning-over motion
of the Kaleidocycle.

Kaleidocycles exhibit interesting properties and pose some topological and geometrical questions.
In \S \ref{sec:final} we indicate some directions of further study to close this exposition.

\medskip 
We list some more preceding works in different fields which are relevant to our topic in some ways.

Mobility analysis of a linkage mechanism studies how many degrees of freedom a particular state of
the linkage has, which corresponds to determination of the local dimension at a point in the
configuration space (see, for example, \cite{Muller2016}).  On the other hand, rigidity of linkages
consisting of hinges are studied in the context of the body-hinge framework (see, for example,
\cite{jordan,katoh}).  The main focus of the study is to give a characterisation for a generic
linkage to have no mobility.  That is, the question is to see if the configuration space is homeomorphic to a
point or isolated points.

Sato and Tanaka \cite{Sato-Tanaka} study the motion of a certain linkage mechanism with a constrained degree of freedom
and observed that soliton solutions appear.

Closed (continuous) curves of constant torsion have attracted sporadic interest of geometers, e.g.,
\cite{bates2013,Calini-Ivey1996, Ivey, Weiner1974,Weiner1977}. 
In particular, \cite{Calini-Ivey1999}
discusses an evolution of a constant torsion curve governed by a sine-Gordon equation in the
continuous setting. 

\section{A mathematical model of linkage}
The purpose of this section is to set up a general mathematical model of linkages. This section is
almost independent of later sections, and can be skipped if the reader is concerned only with our
main results on the motion of Kaleidocycles.
\subsection{A group theoretic model of linkage}\label{sec:linkage-graph}
We define an abstract linkage as a decorated graph,
and its realisation as a certain imbedding of the graph in a Euclidean space.
Our definition generalises the usual graphical model of a pin joint linkage
to allow different types of joints.

Denote by $\mathrm{\it SO}(n)$ the group of orientation preserving linear isometries of 
the $n$-dimensional Euclidean space $\R^n$.
An element of $\mathrm{\it SO}(n)$ is identified with a sequence of $n$-dimensional column vectors 
$[ f_1,f_2,\ldots,f_n ]$
which are mutually orthogonal and have unit length with respect to the standard inner product 
$\langle x, y \rangle$ of $x,y\in \R^n$.
Denote by $\mathrm{\it SE}(n)$ the group of $n$-dimensional orientation preserving Euclidean transformations.
That is, it consists of the affine transformations $\R^n\to \R^n$ which preserves the orientation and the standard metric.
We represent the elements of $\mathrm{\it SE}(n)$ by $(n+1)\times (n+1)-$homogeneous matrices acting on 
\[
  \R^n \simeq \{ {}^t(x_1,x_2,\ldots,x_n,1)\in \R^{n+1}\}
\]
by multiplication from the left.
For example, an element of $\mathrm{\it SE}(3)$ is represented by a matrix
\[
 \begin{pmatrix}
 a_{11} & a_{12} & a_{13} & l_1 \\
 a_{21} & a_{22} & a_{23} & l_2 \\
 a_{31} & a_{32} & a_{33} & l_3 \\
 0 & 0 & 0 & 1 
 \end{pmatrix}.
\]
The vector $l={}^t(l_1,l_2,l_3)$ is called the translation part.
The upper-left $3\times 3$-block of $A$ is called the linear part and denoted by $\bar{A}\in \mathrm{\it SO}(3)$.
Thus, the action on $v\in \R^3$ is also written by $v\mapsto \bar{A}v + l$.

\begin{definition}
An $n$-dimensional \emph{abstract linkage} $L$ consists of the following data:
\begin{itemize}
  \item a connected oriented finite graph $G=(V,E)$ 
  \item a subgroup $J_v\subset \mathrm{\it SE}(n)$ assigned to each $v\in V$, which defines the joint symmetry
  \item an element $C_e\in \mathrm{\it SE}(n)$ assigned to each $e\in E$, which defines the link constraint.
\end{itemize}
\end{definition}
In practical applications, we are interested in the case when $n=2$ or $3$.
When $n=2$ linkages are said to be \emph{planar}, and when $n=3$ linkages are said to be \emph{spatial}.

We say a linkage $L$ is \emph{homogeneous} if for any pair $v_1,v_2\in V$, the following conditions are satisfied:
\begin{itemize}
\item there exists a graph automorphism which maps $v_1$ to $v_2$
(i.e., $\mathrm{\it Aut}(G)$ acts transitively on $V$),
\item $J_{v_1}=J_{v_2}$,
\item and $C_{e_1}=C_{e_2}$ for any $e_1,e_2\in E$.
\end{itemize}

A \emph{state} or \emph{realisation} $\phi$ of an abstract linkage $L$ is an assignment 
of a coset to each vertex
\[
\phi: v\mapsto \mathrm{\it SE}(n)/J_v
\]
such that
for each edge $e=(v_1,v_2)\in E$, the following condition is satisfied:
\begin{equation}\label{eq:edge}
\phi(v_2) J_{v_2}\cap \phi(v_1) J_{v_1} C_e \neq \emptyset,
\end{equation}
where cosets are identified with subsets of $\mathrm{\it SE}(n)$.

Let us give an intuitive description of \eqref{eq:edge}.  Imagine a reference joint sitting at the
origin in a reference orientation.  The subset $\phi(v_1)J_{v_1}$ consists of all the rigid
transformations which maps the reference joint to the joint at $v_1$ with a specified position and
an orientation $\phi(v_1)$ up to the joint symmetry $J_{v_1}$.  The two subsets $\phi(v_2)J_{v_2}$
and $\phi(v_1)J_{v_1} C_e $ intersects if and only if the joint at $v_1$ can be aligned to that at
$v_2$ by the transformation $C_e$.

\begin{ex}
The usual pin joints $v_1,v_2$ connected by a bar-shaped link $e$ of length $l$ are represented by
$J_{v_1}=J_{v_2}=\mathrm{\it SO}(n)$ and $C_e$ being any translation by $l$.  Note that
$\mathrm{\it SE}(3)/J_{v_1}\simeq \R^3$.  It is easy to see that \eqref{eq:edge} amounts to
saying the difference in the translation part of $\phi(v_2)$ and $\phi(v_1)$ should have the norm
equal to $l$.

Two revolute joints (hinges) $v_1,v_2$ in $\R^3$ connected by a link $e$ of length $l$ making an
angle $\alpha$ are represented by $J_{v_1}=J_{v_2}$ being the group generated by rotations around
the $z$-axis and the $\pi$-rotation around the $x$-axis, and $C_e$ being the rotation by $\alpha$
around $x$-axis followed by the translation along $x$-axis by $l$; that is
\[
 J_{v_1}=J_{v_2}=
 \left\{ \begin{pmatrix}
   \cos\theta & \mp \sin\theta & 0 & 0 \\
   \sin\theta & \pm \cos\theta & 0 & 0 \\
   0 & 0& \pm 1 & 0 \\
   0 & 0 & 0 & 1
     \end{pmatrix} \middle| \theta\in \R \right\},
     \quad
  C_e=
  \begin{pmatrix}
   1  & 0 & 0 & l \\
   0 & \cos\alpha & -\sin\alpha & 0 \\
   0 & \sin\alpha & \cos\alpha & 0 \\
   0 & 0 & 0 & 1
     \end{pmatrix}.
\]
Note that $\mathrm{\it SE}(3)/J_{v_1}$ is the space of based lines (i.e., lines with specified
origins) in $\R^3$, and the line is identified with the axis of the hinge.
\end{ex}

The space $\overline{\mathcal{C}}(L)$ of all realisations of a given linkage $L$ admits an action of
$\mathrm{\it SE}(n)$ defined by $\phi\mapsto g\phi(v)$ for $g\in \mathrm{\it SE}(n)$.  The quotient
of $\overline{\mathcal{C}}(L)$ by $\mathrm{\it SE}(n)$ is denoted by $\mathcal{C}(L)$ and called the
\emph{configuration space} of $L$.  Each connected component of $\mathcal{C}(L)$ corresponds to the
mobility of the linkage $L$ in a certain state.  When a connected component is a manifold, its
dimension is what mechanists call \emph{the (internal) degrees of freedom} (DOF, for short).  Given
a pair of points on $\mathcal{C}(L)$, the problem of finding an explicit path connecting the points
is called \emph{motion planning} and has been one of the main topics in mechanics \cite{lavalle}.
In a similar manner, many questions about a linkage can be phrased in terms of the topology and the
geometry of its configuration space.

\begin{ex}
Consider the following spatial linkages consisting of pin joints depicted in Figure \ref{fig:wall}.
In the latter, we assume the two joints $a$ and $b$ are fixed to the wall.  Up to the action of the
global rigid transformation $\mathrm{\it SE}(3)$, these two linkages are equivalent and share the
same configuration space $\mathcal{C}(L)$; in the left linkage, the global action is killed by
fixing the positions of three joints except for $p$.
\begin{figure}[ht]
  \center
\includegraphics[width=4cm]{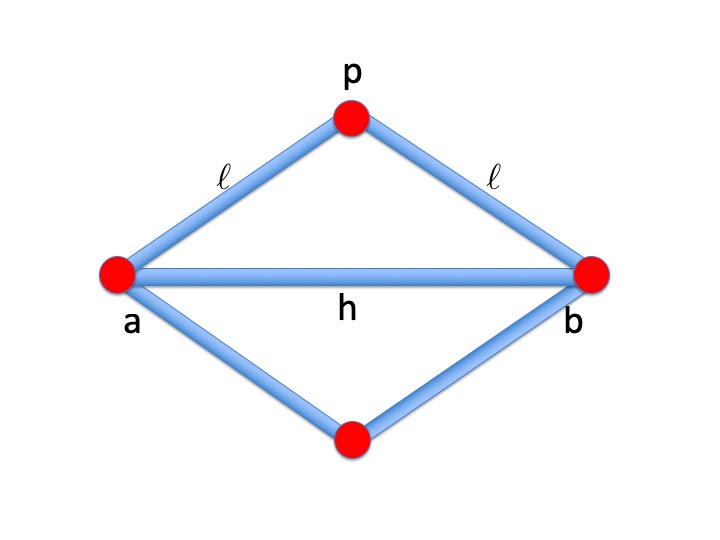}\hspace{1cm}
\includegraphics[width=4cm]{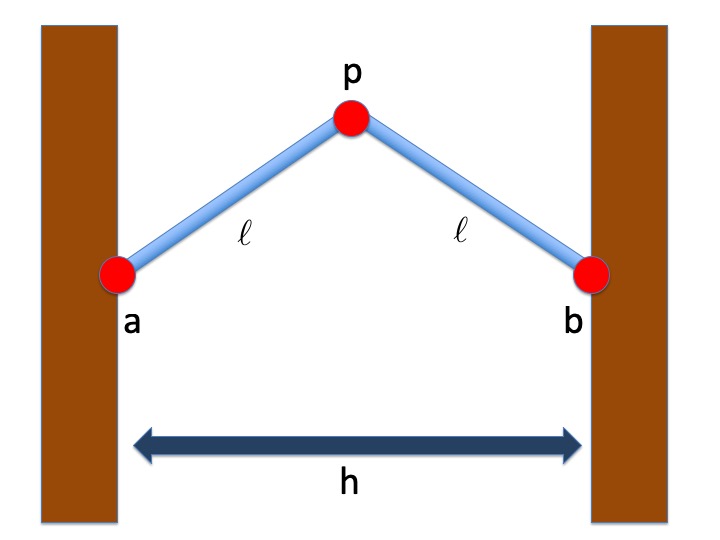}
\caption{Example of equivalent pin joint linkages.}\label{fig:wall}
\end{figure}
The topology of $\mathcal{C}(L)$ changes with respect to the parameter $l$ which is the length of the bars.
Namely, we have
\[
\mathcal{C}(L)=\{ x_p\in \R^3 \mid |x_p-x_a|^2=|x_p-x_b|^2=l^2 \}=
\begin{cases} S^1 & (l>2h) \\ pt & (l=2h) \\ \emptyset & (l<2h) \end{cases}.
\]
This seemingly trivial example is indeed related to a deeper and subtle question on the topology of
the configuration space; the space is identified with the \emph{real} solutions to a system of
algebraic equations.
\end{ex}

\subsection{Hinged linkage in three space}\label{sec:hinged-linkage}
Now, we focus on a class of spatial linkages consisting of hinges,
known also as three dimensional body-hinge frameworks \cite{jordan}.
In this case, the definition in the previous section can be reduced
to a simpler form.

Notice that in $\R^3$ a pair of hinges connected by a link can be modelled by a tetrahedron.  A
hinge is an isometrically embedded real line in $\R^3$.  
Given a pair of
hinges, unit-length segments on the hinges containing the base points in the centre span a
tetrahedron, or a quadrilateral when the two hinges are parallel (see Fig. \ref{fig:tetrahedron} Left).
It is sometimes convenient to decompose the link constraint $C_{(v_1,v_2)}\in \mathrm{\it SE}(3)$ into
three parts; a translation along the hinge direction at $v_1$, a screw motion along an axis
perpendicular to both hinges, and a translation along the hinge direction at $v_2$.  This
corresponds to a common presentation among mechanists called the Denavit--Hartenberg parameters
\cite{DH}.  We can find the decomposition geometrically as follows: Find a line segment which is
perpendicular to the both hinges connected by the link $e$, which we call the \emph{core segment}.
It is unique unless the hinges are parallel.  The intersection points of the core segment and the
hinges are called the \emph{marked} points.  Form a tetrahedron from the line segments on hinges
containing the marked points in the centre.  By construction, this tetrahedron has a special shape
that the line connecting the centre of two hinge edges (the core segment) is perpendicular to the
hinge edges.  Such a tetrahedron is called a \emph{disphenoid}.  The shape of the disphenoid defines
a screw motion along the core segment up to a $\pi$-rotation.  The translations along the hinge
directions are to match the marked points to the base points (see Fig. \ref{fig:tetrahedron}).  To
sum up, a spatial hinged linkage can be considered as a collection of lines connected by disphenoids
at marked points.
\begin{figure}[ht]
\center
\includegraphics[width=4cm]{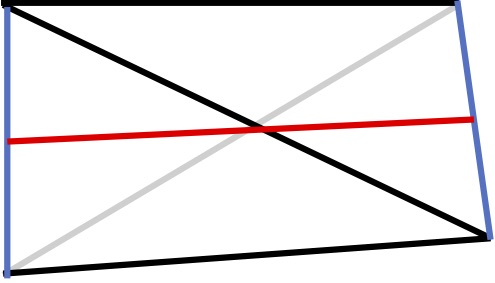} \hspace{1cm}
\includegraphics[width=4cm]{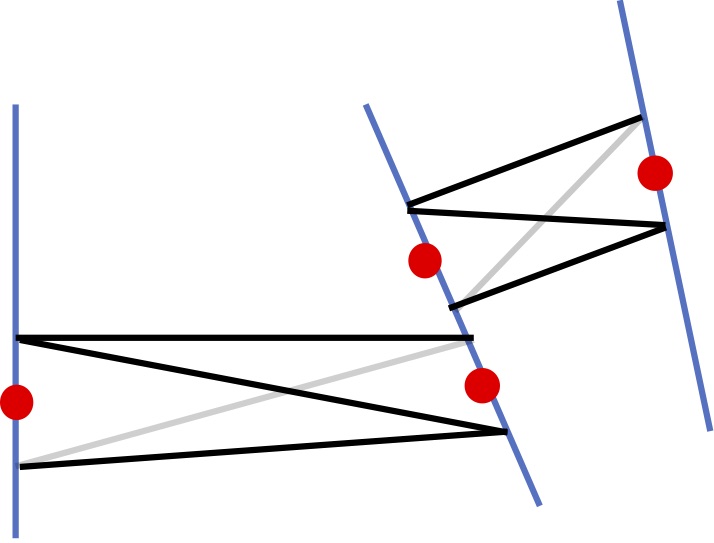}
\caption{Left: a disphenoid formed by two hinge edges,
Right: three hinges connected by disphenoids. The dots indicate the marked points.}
\label{fig:tetrahedron}
\end{figure}
Thus, we arrive in the following definition.
\begin{definition}\label{dfn:hinged-network}
A \emph{hinged network} consists of 
\begin{itemize}
\item a connected oriented finite graph $G=(V,E)$,
\item two edge labels $\nu: E\to [0,\pi)$ called the \emph{torsion angle} and $\varepsilon: E\to
\R_{\ge 0}$ called the \emph{segment length},
\item and a vertex label $\iota_v: E(v)\to \R$ called the \emph{marking}, where $E(v)\subset E$ is
      the set of edges adjacent to $v\in V$.
\end{itemize}
A state of a hinged network is an assignment to each vertex $v\in V$
of an isometric embedding $h_v: \R \to \R^3$ such that for any $(v_1,v_2)\in E$
\begin{enumerate}
\item $|l|=\varepsilon(v_1,v_2)$, where $l=h_{v_1}\circ\iota_{v_1}(v_1,v_2)-h_{v_2}\circ \iota_{v_2}(v_1,v_2)$
\item $l\perp h_{v_1}(\R)$ and $l\perp h_{v_2}(\R)$
\item $\angle h_{v_1}(\R)h_{v_2}(\R)=\nu$, where the angle is measured in the left-hand screw manner with respect to $l$.
\end{enumerate}
Intuitively, $h_v(\R)$ is the line spanned by the hinges, and the first two conditions demand that
the marked points are connected by the core segments $l$, whereas the last condition dictates the
torsion angle of adjacent hinges $h_{v_1}(\R)$ and $h_{v_2}(\R)$.
\end{definition}

A hinged network is said to be \emph{serial} when the graph $G$ is a line graph; i.e., a connected
graph of the shape $\bullet \to \bullet \to \bullet \to \cdots \to \bullet$.  It is said to be
\emph{closed} when the graph $G$ is a circle graph; i.e., a connected finite graph with every vertex
having outgoing degree one and incoming degree one.  A hinged network is homogeneous if
\begin{itemize}
\item $\mathrm{\it Aut}(G)$ acts on $G$ transitively,
\item $\nu(e),\varepsilon(e)$, and $\iota_v$ do not depend on $e\in E$ and $v\in V$.
That is, it is made of congruent tetrahedral links.
\end{itemize}

\begin{ex}
A planar pin joint linkage is a special type of hinged network with $\nu(e)=0$ for all $e\in E$ and
$\iota_v=0$ for all $v\in V$.  That is, all hinges are parallel and marked points are all at the
origin.  On the other hand, any hinged network can be thought of as a spatial pin joint linkage by
replacing every tetrahedral link with four bar links connected by four pin joints forming the
tetrahedron.  Therefore, hinged networks form an intermediate class of linkages which sits between
planar pin joint linkages and spatial pin joint linkages.
\begin{figure}[ht]
  \center
\includegraphics[width=4.5cm]{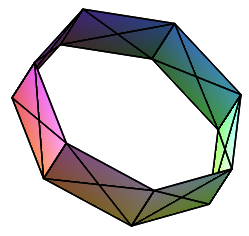}
\caption{A degenerate hinged network over a circle corresponding to a planar six-bar pin joint linkage.}\label{fig:rope}
\end{figure}
\end{ex}

\begin{ex}
The hinged network depicted in Fig. \ref{fig:rope2} is over the wedge sum of two circle graphs.
It exhibits a jump roping motion.
A similar but more complex network is found in \cite[\S 6]{Chen2011}.
\begin{figure}[ht]
  \center
\includegraphics[width=4.5cm]{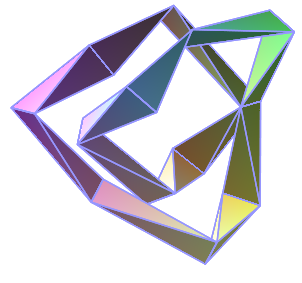}
\caption{A hinged network over the wedge of two circles.}\label{fig:rope2}
\end{figure}
\end{ex}
\begin{ex}
Closed hinged networks with $\varepsilon(e)=0$ (that is, adjacent hinge lines intersect) for all $e\in E$ provide
a linkage model for discrete \emph{developable strips} studied recently by
K. Naokawa and C. M\"uller (see Fig. \ref{fig:developable}).
They are made of (planar) quadrilaterals joined together by
the pair of non-adjacent edges as hinges.
\begin{figure}[ht]
  \center
\includegraphics[width=4.5cm]{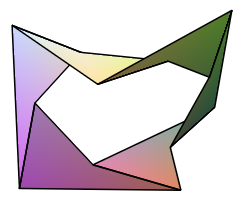}
\includegraphics[width=4.5cm]{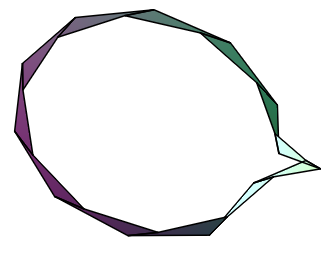}
\caption{Developable discrete M\"obius strip consisting of 6 (respectively 12) congruent
quadrilateral links.}
\label{fig:developable}
\end{figure}
\end{ex}

\section{Hinged network and discrete space curve}\label{circle-system}
In this section, we describe a connection between spatial closed hinged networks and discrete closed
space curves.  This connection is the key idea of this chapter which provides a way to study certain linkages using
tools in discrete differential geometry.

First, we briefly review the basic formulation of discrete space curves (see, for example, \cite{inoguchi2014discrete}).
A \emph{discrete space curve} is a map
\[
\gamma : \mathbb{Z} \rightarrow {\mathbb{R}}^{3}, \quad (i \mapsto \gamma_n).
\]
For simplicity, in this chapter we always assume that $\gamma_n \neq \gamma_{n+1}$ for any $n$ and
that three points $\gamma_{-1}$, $\gamma_{0}$ and $\gamma_{1}$ are not colinear.  The {\it tangent
vector} $T: \Z\to S^2$ is defined by
\begin{equation}\label{def:tangent}
  T_{n} = \frac{\gamma_{n+1}-\gamma_{n}}{\varepsilon_{n}}, \quad
  \varepsilon_{n} = \left| \gamma_{n+1}-\gamma_{n} \right|.
\end{equation}
We say $\gamma$ has a constant speed of $\varepsilon$ if $\varepsilon_{n} = \varepsilon$ for all
$n$.  A discrete space curve with a constant speed is sometimes referred to as an {\it arc length
parametrised curve} \cite{Hoffmann:MI}.  The {\it normal vector} $N: \Z\to S^2$ and the {\it
binormal vector} $B: \Z\to S^2$ are defined by
\begin{align}\label{def:normal_binormal}
&  B_{n} = \begin{cases} \frac{T_{n-1} \times T_{n}}{\left| T_{n-1} \times T_{n} \right|} & (T_{n-1} \times T_{n}\neq0)\\
  B_{n-1} & (T_{n-1} \times T_{n}=0 \text{ and } n>0) \\
  B_{n+1} & (T_{n-1} \times T_{n}=0 \text{ and } n<0),
  \end{cases} \\[2mm]
&  N_{n} = B_{n} \times T_{n},
\end{align}
respectively. Then, $\left[ T_n, \, N_n, \, B_n \right] \in \mathrm{\it SO}(3)$ is called the {\it
Frenet frame} of $\gamma$.  For our purpose, it is more convenient to use a modified version of the
ordinary Frenet frame, which we define as follows.  Set $b_0 = B_0$ and define $b_n = \pm B_n$
recursively so that $\langle b_{n}\times b_{n-1}, T_{n-1} \rangle \ge 0$ and $\langle b_{n-1}, b_n
\rangle\neq -1$.  Then, $\Phi_{n} = [T_n,\widetilde{N}_n,b_n]\in \mathrm{\it SO}(3)$, where
$\widetilde{N}_n=b_n\times T_n$ (see Fig. \ref{fig:frenet1}).
\begin{figure}[ht]
  \center
\includegraphics[width=12cm]{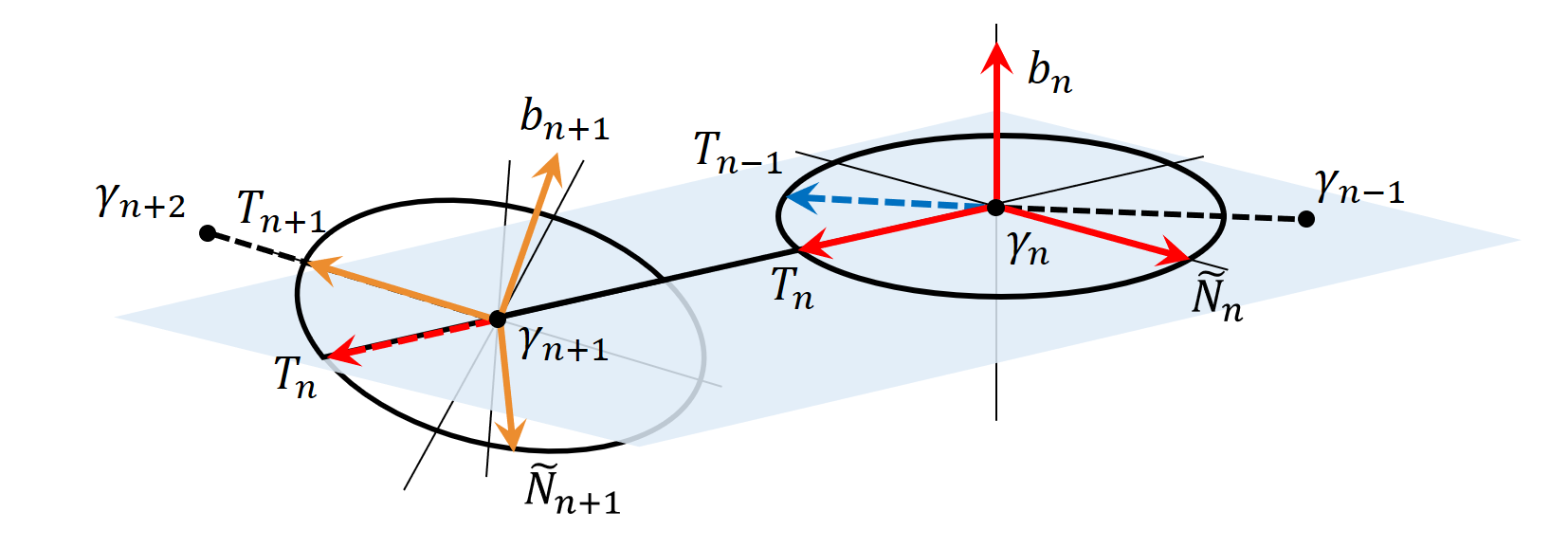}
\caption{A discrete space curve with the frame $\Phi_n$.}\label{fig:frenet1}
\end{figure}

For $\theta\in \R$, we define $R_{1}(\theta) ,R_{3}(\theta) \in \mathrm{\it SO}(3)$ by
\begin{equation}
  R_{1}(\theta) =
    \begin{bmatrix}
      1 & 0 & 0 \\
      0 & \cos{\theta} & -\sin{\theta} \\
      0 & \sin{\theta} & \cos{\theta}
    \end{bmatrix}, \quad
  R_{3}(x) =
    \begin{bmatrix}
      \cos{\theta} & -\sin{\theta} & 0 \\
      \sin{\theta} & \cos{\theta} & 0 \\
      0 & 0 & 1
    \end{bmatrix}.
\end{equation}
There exist $\kappa : \mathbb{Z} \rightarrow [-\pi, \, \pi)$ and $\nu : \mathbb{Z} \rightarrow [0,\, \pi)$ 
such that
\begin{equation}\label{def:discrete_frenet_eq}
  \Phi_{n+1} = \Phi_{n} L_{n}, \quad
  L_{n} = R_{1}(-\nu_{n+1}) R_{3}(\kappa_{n+1}).
\end{equation}
We call $\kappa$ the
\emph{signed curvature angle} and $\nu$ the \emph{torsion angle}.  Fig. \ref{fig:frenet2}
illustrates how to obtain $\Phi_{n-1}$ from $\Phi_{n}$ by \eqref{def:discrete_frenet_eq}.  Note that
we have
\begin{equation}\label{def:angle}
\begin{split}
 \langle T_{n},T_{n-1} \rangle = \cos{\kappa_n}, \quad
&  \langle b_{n},b_{n-1} \rangle = \cos{\nu_n}, \quad
  \langle b_{n},\widetilde{N}_{n-1} \rangle = \sin{\nu_n}, \\[2mm]
& \langle b_n,T_{n} \rangle=\langle b_{n+1},T_{n} \rangle=0.
\end{split}
 \end{equation}
\begin{figure}[ht]
  \center
\includegraphics[width=12cm]{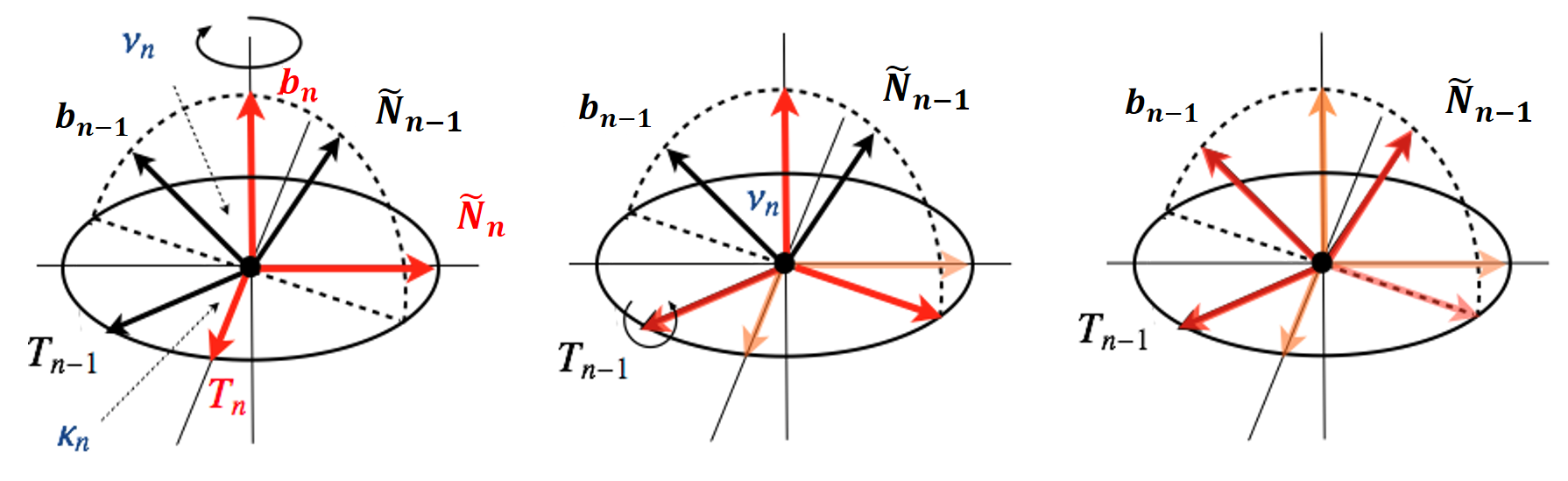}
\caption{The curvature angle $\kappa$ and the torsion angle $\nu$.}\label{fig:frenet2}
\end{figure}

The reason why we introduce the modified frame is that the ordinary Frenet frame behaves
discontinuously under deformation when the ordinary curvature angle vanishes at a point.  During the
turning-over motion of a Kaleidocycle, it goes through such a state at some points, and the
above modified frame behaves consistently even under the situation.

\medskip Fix a natural number $N$.  A discrete space curve $\gamma$ is said to be \emph{closed} of
length $N$ if $\gamma_{n+kN}=\gamma_n$ for any $k\in \Z$.  Unlike the ordinary Frenet frame,
closedness does not imply $\Phi_{n+kN}=\Phi_n$ but they may differ by rotation by $\pi$ around $T_n$.  We say
$b$ is oriented (resp. anti-oriented) if $b_n=b_{n+N}$ (resp. $b_n=-b_{n+N}$) for all $n$.

We can consider a discrete version of the Darboux form \cite{darboux1917,Weiner1977},
which gives a correspondence between spherical curves and space curves.
Given $b: \Z \to S^2$ with $b_n\times b_{n-1}\neq 0$ for all $n$ and $\varepsilon: \Z\to \R_{\ge 0}$,
we can associate a discrete space curve satisfying
\begin{equation}\label{eqn:Darboux_form}
\gamma_0=0,\quad \gamma_n=\gamma_{n-1}+\varepsilon_{n-1}\dfrac{b_{n} \times b_{n-1}}{|b_{n} \times b_{n-1}|}, 
\end{equation}
which we denote by $\gamma^{b,\varepsilon}$.
The curve $\gamma^{b,\varepsilon}$ is closed of length $N$ if
\begin{equation}
\sum_{n=0}^{N-1} \left( \varepsilon_{k+n} \dfrac{b_{k+n+1} \times b_{k+n}}{|b_{k+n+1} \times b_{k+n}|} \right) =0 
\end{equation}
for all $k$.

Notice that a serial (resp. closed) hinged network with $\iota_v= 0$ for all $v\in V$ (see
Def. \ref{dfn:hinged-network}) can be modelled by an open (resp. a closed) discrete space curve; its
base points form the curve and hinge directions are identified with $b_n$ (see
Fig. \ref{fig:frenet0}).  This is the crucial observation of this chapter.
\begin{figure}[ht]
  \center
\includegraphics[width=10cm]{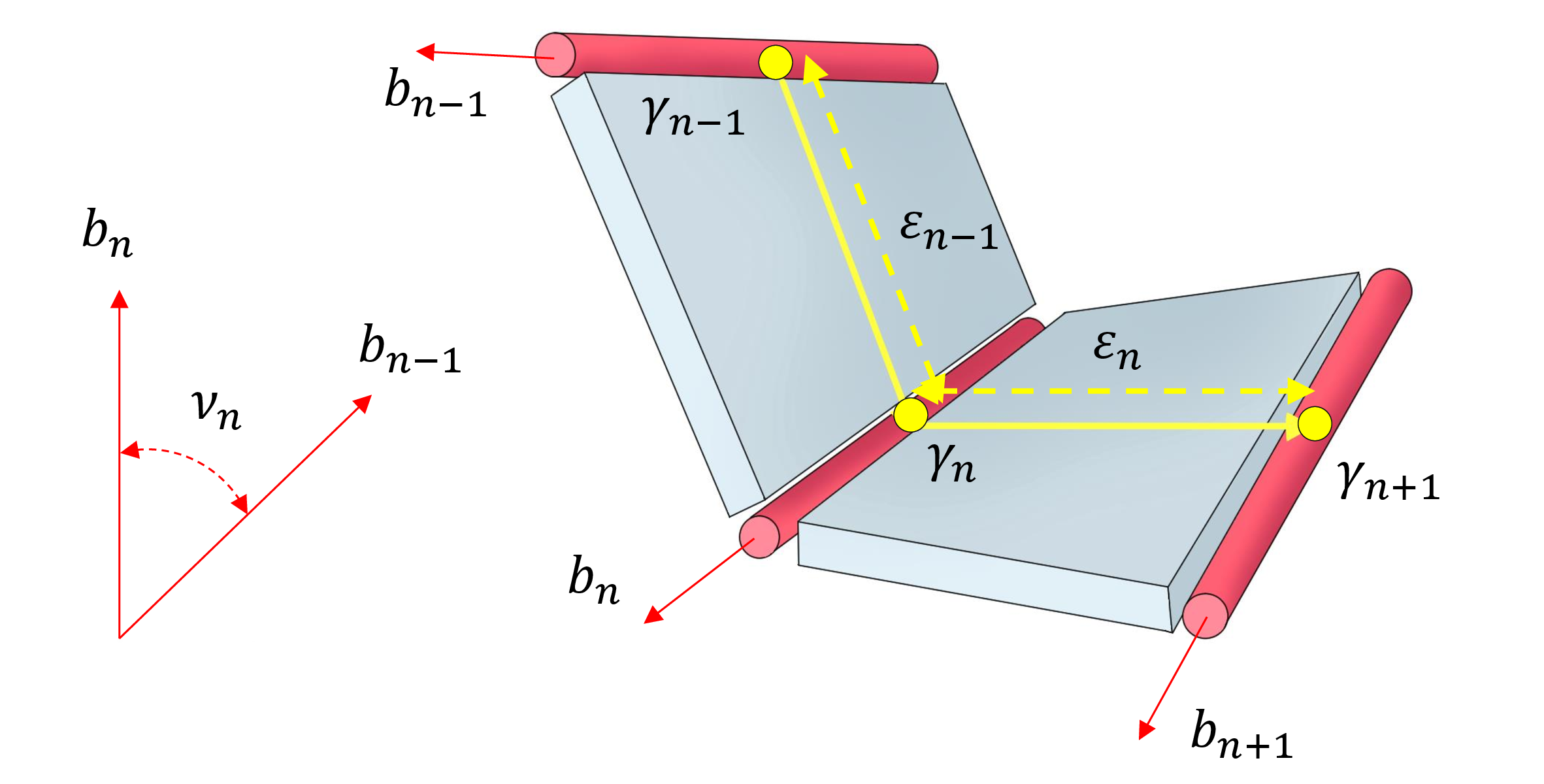}
\caption{Hinged network and discrete space curve.}\label{fig:frenet0}
\end{figure}

Now we introduce our main object, \emph{Kaleidocycles}, which are homogeneous closed hinged networks.
We model them as constant speed discrete space curves of constant torsion.
They are a generalisation to a popular paper toy called the Kaleidocycle (see, e.g.,
\cite{Byrnes,Schattschneider1987}).  A serial hinged network similar to our Kaleidocycle is proposed
in \cite{Moses2013}.

\begin{definition}\label{dfn:kaleidocycle}
Fix $\nu\in [0,\pi]$ and $\epsilon>0$.  An $N$-Kaleidocycle with a speed $\varepsilon$ and a torsion
angle $\nu$ is a closed discrete space curve $\gamma$ of length $N$ which has a constant speed
$\varepsilon_n=\varepsilon$ and a constant torsion angle $\nu_n=\nu$.  It is said to be
\emph{oriented} (resp. \emph{anti-oriented}) when associated $b$ is oriented (resp. anti-oriented).
\end{definition}
When $\nu$ is either $0$ or $\pi$, the corresponding Kaleidocycles are planar,
and we call them \emph{degenerate}.
\begin{figure}[ht]
    \center
  \includegraphics[width=4.5cm]{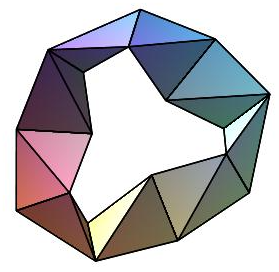}
  \includegraphics[width=4.5cm]{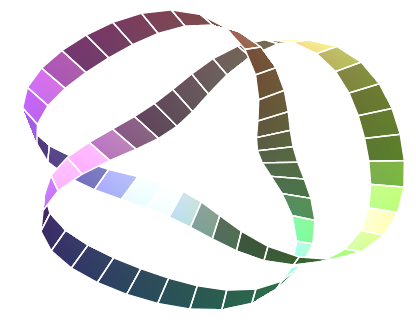}
  \caption{Left: anti-oriented Kaleidocycle with $N=9$.
  Right: a Kaleidocycle with a knotted topology.}\label{fig:kaleido}
\end{figure}
For fixed $N$ and $\varepsilon$,
an oriented (resp. anti-oriented) non-degenerate Kaleidocycle
with a torsion angle $\nu$ is determined by the Darboux form $\gamma^{b,\varepsilon}$
by a map $b: \Z\to S^2$ satisfying
p\begin{itemize}
\item $b_{n+N}=b_n$ (resp. $b_{n+N}=-b_n$),
\item $\langle b_n,b_{n+1}\rangle = \cos\nu$,
\item $\sum\limits_{n=0}^{N-1} b_{n+1} \times b_{n} =0$.
\end{itemize}
We use $b$ and $\gamma$ interchangeably to represent a Kaleidocycle.

Consider the real algebraic variety $\overline{\mathcal{M}}_N$ defined by the following system of
quadratic equations (\cite[Ex. 5.2, 8.13]{Bertini}):
\begin{equation}\label{eq:DH}
\langle b_n, b_{n+1} \rangle = c  \quad (0\le n < N), \qquad
\sum_{n=0}^{N-1} b_{n+1}\times b_{n} = 0, \qquad b_N=\pm b_0,
\end{equation}
where $c$ is considered as an indeterminate.  The orthogonal group $\mathrm{\it O}(3)$ acts on
$b_i$'s in the standard way, and hence, on $\overline{\mathcal{M}}_N$.  Denote by $\mathcal{M}_N$ the quotient
of $\overline{\mathcal{M}}_N$ by the action of $\mathrm{\it O}(3)$.  The variety $\mathcal{M}_N$ serves as the
configuration space of all $N$-Kaleidocycles with varying $c=\cos\nu$. It decomposes into two
disjoint sub-spaces $\mathcal{M}_N^+$ consisting of all oriented Kaleidocycles ($b_N=b_0$) and
$\mathcal{M}_N^-$ consisting of anti-oriented ones ($b_N=-b_0$).

As $\mathcal{M}_N^-$ (resp. $\mathcal{M}_N^+$) is a closed variety, its image under the projection
$\pi_c$ onto the $c$-axis is a union of closed intervals.  Notice that the image
$\pi_c(\mathcal{M}_N^-)$ does not coincide with the whole interval $[-1,1]$; $c=1$ means $b_i$ are
all equal so we cannot have $b_N=-b_0$.  The fibre $\pi_c^{-1}(c)$ consists of $N$-Kaleidocycles
with a fixed $c$.  With a generic value of $c$, a simple dimension counting in \eqref{eq:DH} shows
that $\dim(\pi_c^{-1}(c))=N-6$.  Hence, the degree of freedom (DOF) of the Kaleidocycle with a torsion
angle $\nu=\arccos(c)$ is generally $N-6$.  For $N>6$, a generic Kaleidocycle is
\emph{reconfigurable} meaning that it can continuously change its shape.  We will investigate a
particular series of reconfiguration in the next section.
\begin{remark}
The most popular Kaleidocycle with $N=6$ has $c=0$, which is equivalent to the threefold symmetric
Bricard 6R linkage (Fig. \ref{fig:bricard}).  This Kaleidocycle is highly symmetric and not generic,
resulting in $1$ DOF \cite{Foweler-Guest2005}.
\end{remark}

\section{Deformation of discrete curves}\label{sec:motion}
\subsection{Continuous isoperimetric deformations on discrete curves}\label{sec:isoperimetric}
Kaleidocycles exhibit a characteristic turning-over motion (see Fig. \ref{fig:turning} and see
\cite{code} for some animations).  In general, an $N$-Kaleidocycle has $N-6$ degrees of freedom so
that it wobbles in addition to turning-over.  With special values of torsion angle, however, the
DOF of the Kaleidocycle seems to degenerate to exactly one, leaving only the turning-over motion as we
will discuss in \S \ref{sec:final}.  In this case, the motion of the core segment looks to be
orthogonal to the hinge directions.  In the following, we would like to model the motion explicitly.
It turns out that we can construct the motion of Kaleidocycles using semi-discrete mKdV and sine-Gordon equations.
\begin{figure}[ht]
    \center
  \includegraphics[width=3.0cm]{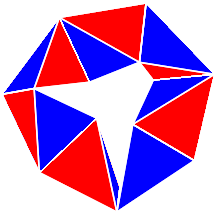}
  \includegraphics[width=3.1cm]{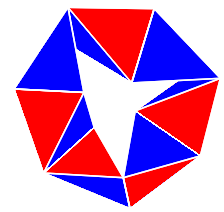}
  \includegraphics[width=2.9cm]{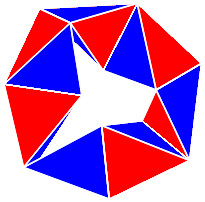}
  \includegraphics[width=2.9cm]{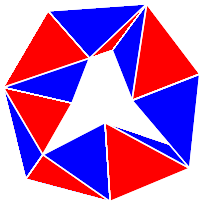}
  \caption{Turning-over motion of a Kaleidocycle with $N=7$.}
  \label{fig:turning}
\end{figure}

In this section, we consider certain continuous deformations of discrete space curves which
correspond to motion of homogeneous serial and closed hinged networks.  Our approach is to
construct a flow on the configuration space by differential-difference equations.  We use the same
notations as in Section \ref{circle-system}.  Observe that a hinged network moves in such a way that
its tetrahedral links are not distorted.  In the language of discrete space curves, the motion
corresponds to a deformation which preserves the speed $\epsilon_n$ and the torsion angle $\nu_n$
for all $n$.

Let $\gamma(0): \Z\to \R^3$ be an (open) discrete space curve which has a constant speed
$\varepsilon_n(0)=\varepsilon_*(0)$ and a constant torsion angle $\nu_n(0)=\nu_*(0)$.  Given a
family of functions $\doubleu(t) : \mathbb{Z} \rightarrow \mathbb{R}$ with the deformation parameter
$t \in \mathbb{R}$ and a constant $\rho>0$, we consider a family of discrete space curves
$\gamma(t)$ defined by
\begin{equation}\label{def:deformation_curve}
  \frac{d \gamma_n}{d t} = \frac{\varepsilon_n}{\rho}
  \left( \cos{\doubleu_n}T_n + \sin{\doubleu_n}\widetilde{N}_n \right) \qquad (n\in \Z).
\end{equation}
That is, the motion of each point $\gamma_{n}$ is confined in the osculating plane and its speed
depends only on the length of the segment $\varepsilon_n=|\gamma_{n+1}-\gamma_{n}|$.  We say a
deformation is \emph{isoperimetric} if the segment length $\varepsilon_n$ does not depend on $t$ for
all $n$, We would like to find conditions on $\doubleu$ under which the above deformation is
isoperimetric.  From \eqref{def:tangent}, \eqref{def:discrete_frenet_eq} and
\eqref{def:deformation_curve}, we have
\begin{align*}
  \frac{d \varepsilon_n}{d t}
  &= \frac{\varepsilon_n}{\rho}
  \left\langle \Phi_{n+1}
  \begin{bmatrix}
    \cos{\doubleu_{n+1}} \\ \sin{\doubleu_{n+1}} \\ 0
  \end{bmatrix} -
  \Phi_{n}
  \begin{bmatrix}
    \cos{\doubleu_{n}} \\ \sin{\doubleu_{n}} \\ 0
  \end{bmatrix},
  \Phi_{n}
  \begin{bmatrix}
    1 \\ 0 \\ 0
  \end{bmatrix}
   \right\rangle \\[0.5pc]
   &= \frac{\varepsilon_n}{\rho}
   \left\langle
   \Phi_{n}
   \begin{bmatrix}
     \cos{(\kappa_{n+1} + \doubleu_{n+1})} - \cos{\doubleu_{n}} \\
     \cos{\nu_n}\sin{(\kappa_{n+1} + \doubleu_{n+1})} - \sin{\doubleu_{n}} \\
     -\sin{\nu_n} \sin{(\kappa_{n+1} + \doubleu_{n+1})}
   \end{bmatrix},
   \Phi_{n}
   \begin{bmatrix}
     1 \\ 0 \\ 0
   \end{bmatrix}
    \right\rangle \\[0.5pc]
    &= \frac{\varepsilon_n}{\rho}
    \big( \cos{(\kappa_{n+1} + \doubleu_{n+1})} - \cos{\doubleu_{n}} \big).
\end{align*}
Therefore, for each $n$, ${d \varepsilon_n}/{d t}$ vanishes if and only if
\begin{equation}\label{eqn:kappa_psi_cos}
  \cos{(\kappa_{n+1} + \doubleu_{n+1})} - \cos{\doubleu_{n}} = 0,
\end{equation}
which yields
\begin{equation}\label{condition:kappa_psi2}
    \doubleu_n = - \doubleu_{n-1}-\kappa_n,
\end{equation}
or
\begin{equation}\label{condition:kappa_psi1}
     \doubleu_n = \doubleu_{n-1}- \kappa_n.
\end{equation}
We consider a deformation when \eqref{condition:kappa_psi2} (resp. \eqref{condition:kappa_psi1})
simultaneously holds for all $n$.  Note that in this case $\doubleu_n(t)$ for all $n$ is determined
once $\doubleu_0(t)$ is given.

Those deformations are characterised by the following propositions:
\begin{proposition}\label{prop1}
Let $\gamma(0): \Z\to \R^3$ be a discrete space curve with a constant speed
$\varepsilon_n(0)=\varepsilon_*(0)$ and a constant torsion angle $\nu_n(0)=\nu_*(0)$.  Let
$\gamma(t)$ be its deformation according to \eqref{def:deformation_curve} with $\doubleu: \Z \to \R$
satisfying the condition \eqref{condition:kappa_psi2}. Then we have:
\begin{enumerate}
 \item The speed $\varepsilon_n(t)$ and the torsion angle $\nu_n(t)$ do not depend on $t$ nor $n$.
       That is, $\varepsilon_n(t)=\varepsilon_*(0)$ and $\nu_n(t)=\nu_*(0)$ for all $t$ and $n$.
 \item The signed curvature angle $\kappa_n=\kappa_n(t)$ and $\doubleu_n=\doubleu_n(t)$ satisfy
       \begin{equation}\label{eq:sine-Gordon_kappa2}
         \frac{d \kappa_n}{d t} = {\alpha} \left( \sin{\doubleu_{n-1}} - \sin{\doubleu_{n}} \right),
       \end{equation}
       where ${\alpha} = \frac{1 + \cos{\nu_*(0)}}{\rho}$.
 \item The deformation of the frame $\Phi_n(t) = [T_n(t), \, \widetilde{N}_n(t), \, b_n(t)]$ is given by
       \begin{equation}\label{def:deformation_frame1}
         \begin{split}
         \frac{d \Phi_n}{d t} &= \Phi_n M_n, \\
         M_n &=\frac{1}{\rho}
             \begin{bmatrix}
               0 & \left( 1 + \cos{\nu_*(0)} \right)\sin{\doubleu_n} & - \sin{\nu_*(0)}\sin{\doubleu_n} \\[0.3pc]
               - \left( 1 + \cos{\nu_*(0)} \right)\sin{\doubleu_n} & 0 & \sin{\nu_*(0)}\cos{\doubleu_n} \\[0.3pc]
                  \sin{\nu_*(0)}\sin{\doubleu_n} & - \sin{\nu_*(0)}\cos{\doubleu_n} & 0
                \end{bmatrix}.
              \end{split}
          \end{equation}
  \end{enumerate}
\end{proposition}
\begin{proposition}\label{prop2}
Let $\gamma(0): \Z\to \R^3$ be a discrete space curve with a constant speed
$\varepsilon_n(0)=\varepsilon_*(0)$ and a constant torsion angle $\nu_n(0)=\nu_*(0)$.
Let $\gamma(t)$ be its deformation according to \eqref{def:deformation_curve}
with $\doubleu: \Z \to \R$ satisfying the condition \eqref{condition:kappa_psi1}.
Then we have:
  \begin{enumerate}
    \item The speed $\varepsilon_n(t)$ and the torsion angle $\nu_n(t)$ do not depend on $t$ nor $n$.
          That is, $\varepsilon_n(t)=\varepsilon_*(0)$ and $\nu_n(t)=\nu_*(0)$ for all $t$ and $n$.
    \item The signed curvature angle $\kappa_n=\kappa_n(t)$ and $\doubleu_n=\doubleu_n(t)$ satisfy
          \begin{equation}\label{eq:sine-Gordon_kappa}
            \frac{d \kappa_n}{d t} = - \hat{\alpha} \left( \sin{\doubleu_{n}} + \sin{\doubleu_{n-1}} \right),
          \end{equation}
          where $\hat{\alpha} = \frac{1 - \cos{\nu_*(0)}}{\rho}$.
    \item The deformation of the frame $\Phi_n(t) = [T_n(t), \, \widetilde{N}_n(t), \, b_n(t)]$ is given by
          \begin{equation}\label{def:deformation_frame}
            \begin{split}
            \frac{d {\Phi}_n}{d t} &= \Phi_n M_n, \\
            M_n &=\frac{1}{\rho}
                \begin{bmatrix}
                  0 & \left( 1 - \cos{\nu_*(0)} \right)\sin{\doubleu_n} & \sin{\nu_*(0)}\sin{\doubleu_n} \\[0.3pc]
                  - \left( 1 - \cos{\nu_*(0)} \right)\sin{\doubleu_n} & 0 & -\sin{\nu_*(0)}\cos{\doubleu_n} \\[0.3pc]
                  - \sin{\nu_*(0)}\sin{\doubleu_n} & \sin{\nu_*(0)}\cos{\doubleu_n} & 0
                \end{bmatrix}.
              \end{split}
          \end{equation}
  \end{enumerate}
\end{proposition}
\begin{proof}
We only prove Proposition \ref{prop1} since Proposition \ref{prop2} can be proved in
the same manner.  We first show the second and the third statements.  We denote $\dot{f} = \frac{d
f}{d t}$, $\nu = \nu_*(0)$ and $\varepsilon = \varepsilon_*(0)$ for simplicity.  Since $\varepsilon$
is a constant by the preceding argument, the deformation of $T_n$ can be computed from
\eqref{def:deformation_curve} and \eqref{condition:kappa_psi1} as
  \begin{align}
    \dot{T_n} &=
      \frac{1}{\rho} \Phi_n
      \left( L_n
      \begin{bmatrix}
        \cos{\doubleu_{n+1}} \\ \sin{\doubleu_{n+1}} \\ 0
      \end{bmatrix}
    - \begin{bmatrix}
        \cos{\doubleu_n} \\ \sin{\doubleu_n} \\ 0
      \end{bmatrix}
      \right)\nonumber\\
      &=\frac{1}{\rho} \Phi_n
      \begin{bmatrix}
        \cos{\left( \kappa_{n+1} + \doubleu_{n+1} \right)} - \cos{\doubleu_n} \\
        \cos{\nu} \sin{\left( \kappa_{n+1} + \doubleu_{n+1} \right)} - \sin{\doubleu_n} \\
        -\sin{\nu} \sin{\left( \kappa_{n+1} + \doubleu_{n+1} \right)}
      \end{bmatrix}\nonumber\\
      &=
      \frac{1}{\rho} \Phi_n
      \begin{bmatrix}
        0 \\
        - \left( 1 + \cos{\nu} \right) \sin{\doubleu_n} \\
          \sin{\nu} \sin{\doubleu_n}
      \end{bmatrix}.\label{eq:dotT}
  \end{align}
Differentiating $\cos{\kappa_n} = \langle T_n, T_{n-1} \rangle$ with respect to $t$, we have
\begin{equation}\label{eq:dotkappa}
    -\dot{\kappa_n} \sin{\kappa_n} =
    \langle \dot{T_n}, T_{n-1} \rangle + \langle T_n, \dot{T}_{n-1} \rangle.
\end{equation}
 Noting
  \begin{equation}
    T_{n-1} = \Phi_n L_{n-1}^{-1}
    \begin{bmatrix}
      1 \\
      0 \\
      0
    \end{bmatrix}
    =
    \Phi_n
    \begin{bmatrix}
      \cos{\kappa_n} \\
      - \sin{\kappa_n} \\
      0
    \end{bmatrix},
  \end{equation}
  and
  \begin{align}
    \dot{T}_{n-1} &=
    \frac{1}{\rho} \Phi_n L_n^{-1}
    \begin{bmatrix}
      0 \\
      - \left( 1 + \cos{\nu} \right) \sin{\doubleu_{n-1}} \\
        \sin{\nu} \sin{\doubleu_{n-1}}
    \end{bmatrix}\nonumber\\
    &=
    \frac{1}{\rho} \Phi_n
    \begin{bmatrix}
      - \left( 1 + \cos{\nu} \right) \sin{\kappa_n} \sin{\doubleu_{n-1}} \\
      - \left( 1 + \cos{\nu} \right) \cos{\kappa_n} \sin{\doubleu_{n-1}}\\
      - \sin{\nu} \sin{\doubleu_{n-1}}
    \end{bmatrix},
  \end{align}
we get from \eqref{eq:dotT} and \eqref{eq:dotkappa}
  \begin{equation}\label{eq:sine_gordon1}
    \dot{\kappa_n} = \frac{1 + \cos{\nu}}{\rho}
    \left( \sin{\doubleu_{n-1}} - \sin{\doubleu_{n}} \right),
  \end{equation}
which is equivalent to \eqref{eq:sine-Gordon_kappa}.
This proves the second statement.
Next, we see from the definition of $b_n$
  \begin{equation}\label{eq:dotB}
    \dot{b}_n =
    \frac{d}{d t}\left( \frac{1}{\left| T_{n-1} \times T_n \right|}\right) T_{n-1} \times T_n
    + \frac{1}{\left| T_{n-1} \times T_n \right|}
    \left(
    \dot{T}_{n-1} \times T_n + T_{n-1} \times \dot{T}_n
    \right).
  \end{equation}
  Noting
  \begin{equation}
    T_{n-1} \times T_n = \Phi_n
    \begin{bmatrix}
      \cos{\kappa_n} \\
      - \sin{\kappa_n} \\
      0
    \end{bmatrix}
    \times
    \Phi_n
    \begin{bmatrix}
      1 \\
      0 \\
      0
    \end{bmatrix}
    =
    \Phi_n
    \begin{bmatrix}
      0 \\
      0 \\
      \sin{\kappa_n}
    \end{bmatrix},
  \end{equation}
  \begin{align}
    \dot{T}_{n-1} \times T_n &=
    \frac{1}{\rho} \Phi_n
    \begin{bmatrix}
      - \left( 1 + \cos{\nu} \right) \sin{\kappa_n} \sin{\doubleu_{n-1}} \\
      - \left( 1 + \cos{\nu} \right) \cos{\kappa_n} \sin{\doubleu_{n-1}}\\
      - \sin{\nu} \sin{\doubleu_{n-1}}
    \end{bmatrix}
    \times
    \Phi_n
    \begin{bmatrix}
      1 \\
      0 \\
      0
    \end{bmatrix}\nonumber\\
    &=
    \frac{1}{\rho} \Phi_n
    \begin{bmatrix}
      0 \\
      - \sin{\nu} \sin{\doubleu_{n-1}} \\
       \left( 1 + \cos{\nu} \right) \cos{\kappa_n} \sin{\doubleu_{n-1}}
    \end{bmatrix},
  \end{align}
  and
  \begin{align}
    T_{n-1} \times \dot{T}_n &=
    \Phi_n
    \begin{bmatrix}
      \cos{\kappa_n} \\
      -\sin{\kappa_n} \\
      0
    \end{bmatrix}
    \times
    \frac{1}{\rho} \Phi_n
    \begin{bmatrix}
      0 \\
      - \left( 1 + \cos{\nu} \right) \sin{\doubleu_n} \\
        \sin{\nu} \sin{\doubleu_n}
    \end{bmatrix}\nonumber\\
    &=
    \frac{1}{\rho} \Phi_n
    \begin{bmatrix}
      - \sin{\nu} \sin{\kappa_n} \sin{\doubleu_n} \\
      - \sin{\nu} \cos{\kappa_n} \sin{\doubleu_n} \\
      - \left( 1 + \cos{\nu} \right) \cos{\kappa_n} \sin{\doubleu_{n}}
    \end{bmatrix},
  \end{align}
  we get from  \eqref{eq:sine_gordon1} and \eqref{eq:dotB}
  \begin{equation}\label{def:dotB}
    \dot{b}_n =
    \frac{1}{\rho} \Phi_n
    \begin{bmatrix}
      - \sin{\nu} \sin{\doubleu_n} \\
        \sin{\nu} \cos{\doubleu_n}\\
      0
    \end{bmatrix}.
\end{equation}
We immediately obtain $\dot{\widetilde{N_n}}$ from \eqref{eq:dotT} and \eqref{eq:dotB} as
\begin{equation}\label{eq:dotN}
  \dot{\widetilde{N}} = \dot{b_n} \times T_n + b_n \times \dot{T_n}
  = \frac{1}{\rho} \Phi_n
    \begin{bmatrix}
      (1 + \cos{\nu}) \sin{\doubleu_n} \\
      0 \\
      - \sin{\nu} \cos{\doubleu_n}
    \end{bmatrix}.
\end{equation}
Then we have \eqref{def:deformation_frame} from \eqref{eq:dotT}, \eqref{eq:dotB} and
\eqref{eq:dotN}, which proves the third statement. Finally, differentiating $\cos{\nu} = \langle b_n
, b_{n-1} \rangle$ with respect to $t$, it follows from \eqref{def:dotB} and
\eqref{eqn:kappa_psi_cos} that
\begin{displaymath}
 -\dot{\nu} \sin{\nu} =
\langle \dot{b_n}, b_{n-1} \rangle + \langle b_n, \dot{b}_{n-1} \rangle
=
- \frac{\sin^2{\nu}}{\rho}\big(\cos{\left( \kappa_n + \doubleu_n \right) - \cos{\doubleu_{n-1}}} \big)
 = 0,
\end{displaymath}
which implies $\dot{\nu} = 0$. This completes the proof of the first statement.
\end{proof}
\begin{remark}
The condition \eqref{condition:kappa_psi2} suggests the {\it potential function} $\theta_n$ in
Proposition \ref{prop1} such that we have
    \begin{equation}
      \kappa_n = \frac{\theta_{n+1} - \theta_{n-1}}{2}, \quad
      \doubleu_n = \frac{\theta_{n} - \theta_{n+1}}{2}.
    \end{equation}
Then, \eqref{eq:sine-Gordon_kappa2} is rewritten as
  \begin{equation}
      \frac{d}{dt} \left( \theta_{n+1}+\theta_n \right) =
      2\alpha \sin{\left( \frac{\theta_{n+1} - \theta_{n}}{2}\right)}.
    \end{equation}
To the best of the authors' knowledge, this is a novel form of the {\it semi-discrete potential mKdV
equation}.  In fact, the continuum limit $\alpha=\frac{2}{\epsilon}$, $X=\epsilon n + t$,
$T=\frac{\epsilon^2}{12}t$, $\epsilon\to 0$ yields the potential mKdV equation
  \begin{equation}
   \theta_T + \frac{1}{2}(\theta_{X})^3 + \theta_{XXX}=0.
  \end{equation}
Similarly, introducing the potential function $\theta_n$ in Proposition \ref{prop2}
such that
  \begin{equation}
    \kappa_n = \frac{\theta_{n+1} - \theta_{n-1}}{2}, \quad
    \doubleu_n = -\frac{\theta_{n+1} + \theta_n}{2},
  \end{equation}
  suggested by \eqref{condition:kappa_psi1}, we can rewrite \eqref{eq:sine-Gordon_kappa} as
\begin{equation}\label{eq:sine_Gordon}
    \frac{d}{dt} \left( \theta_{n+1}-\theta_n \right) =
    2\alpha \sin{\left( \frac{\theta_{n+1} + \theta_{n}}{2}\right)},
\end{equation}
which is nothing but the {\it semi-discrete sine-Gordon equation} \cite{Boiti:dsG,Orfanidis1, Orfanidis2}.
\end{remark}
\begin{remark}
In the above argument, we assume that the speed of the deformation $\rho$ in
\eqref{def:deformation_curve} is a constant and does not depend on $n$.  Then, by demanding that the
deformation preserve arc length (\eqref{condition:kappa_psi2} or \eqref{condition:kappa_psi1}),
it followed that the torsion angle is also preserved.  Conversely, it seems to be the case that for
the deformation to preserve both the arc length and the torsion angle, the speed $\rho$ is required not
to depend on $n$.
\end{remark}

\begin{remark}[Continuum limit]
The isoperimetric torsion-preserving discrete deformations for the discrete space curves of
constant torsion have been considered in \cite{inoguchi2014discrete}, where the deformations are
governed by the discrete sine-Gordon and the discerte mKdV equations.  It is possible to obtain the
continuous deformations discussed in this section by suitable continuum limits from those discrete
deformations. More precisely, let $\gamma_n^m$ ($m\in\mathbb{Z}$) be a family of discrete curves
obtained by applying the discrete deformations $m$ times to $\gamma_n^0=\gamma_n$, where $\gamma_n$ is
the discrete curve with a constant speed $\varepsilon$ and a constant torsion angle $\nu$. Then
the above discrete deformation is given by
\begin{equation}\label{def:deformation_curve_discrete}
 \gamma_n^{m+1} =
\gamma_n^m + \delta_m \left( \cos{\doubleu_n^m} T_n^m + \sin{\doubleu_n^m} N_n^m \right).
\end{equation}
Then if we choose $\delta_m$ and $\doubleu_0^m$ so that the sign of
$\sigma_n^m=\sin{(\doubleu_{n+1}^m+\kappa_{n+1}^m-\doubleu_{n-1}^m)}$ does not depend on
$n$, the isoperimetric condition and the compatibility condition of the Frenet frame
yield the discrete mKdV equation
\begin{equation}\label{eq:discrete_mKdV}
  \frac{\doubleu_{n+1}^{m+1} - \doubleu_n^m}{2}
  = \arctan{\left( \frac{b + a}{b - a} \tan{\frac{\doubleu_n^{m+1}}{2}} \right)}
  - \arctan{\left( \frac{b + a}{b - a} \doubleu_{n+1}^m \right)},
\end{equation}
when $\sigma_n^m>0$, and the discrete sine-Gordon equation
\begin{equation}\label{eq:discrete_sG}
  \frac{\doubleu_{n+1}^{m+1} + \doubleu_n^m}{2}
  = \arctan{\left( \frac{b + a}{b - a} \tan{\frac{\doubleu_n^{m+1}}{2}} \right)}
    + \arctan{\left( \frac{b + a}{b - a} \doubleu_{n+1}^m \right)},
\end{equation}
when $\sigma_n^m<0$  with
\begin{equation}
a = \left( 1 + \tan^2{\frac{\nu}{2}} \right) \varepsilon,\quad
b = \left( 1 + \tan^2{\frac{\nu}{2}} \right) \delta.
\end{equation}
For the discrete mKdV equation \eqref{eq:discrete_mKdV}, in the limit of
\begin{equation}
  a = \frac{2 \varepsilon}{\rho \alpha},\quad m = \frac{\rho}{\varepsilon \delta}  t,\quad
  b \rightarrow 0\ (\delta \rightarrow 0),
\end{equation}
\eqref{eq:discrete_mKdV} is reduced to
the semi-discrete mKdV equation \eqref{eq:sine-Gordon_kappa2}.
Similarly, the discrete sine-Gordon equation \eqref{eq:discrete_sG} is reduced to
the semi-discrete sine-Gordon equation \eqref{eq:sine-Gordon_kappa} in the limit
\begin{equation}
  a = \frac{\alpha \rho}{\varepsilon}, \quad m = \frac{\rho}{\varepsilon \delta}  t,\quad
  b \rightarrow \infty\ (\delta \rightarrow 0).
\end{equation}
Obviously, the discrete deformation equation of the discrete curve
\eqref{def:deformation_curve_discrete} is reduced to the continuous deformation
equation \eqref{def:deformation_curve}. Moreover, it is easily verified that
the discrete deformation equations of the Frenet frame in \cite{inoguchi2014discrete}
are reduced to \eqref{def:deformation_frame1} and \eqref{def:deformation_frame}.
\end{remark}

\subsection{Turning-over motion of Kaleidocycles}\label{sec:motion_kaleidocycle}
An $N$-Kaleidocycle corresponds to a closed discrete curve $\gamma$ of length $N$
having a constant speed $\varepsilon$ and a constant torsion angle $\nu$ whose $b$ is oriented.
Since $\gamma$ is closed, for \eqref{def:deformation_curve} to define a deformation of $\gamma$,
we need a periodicity condition $\doubleu_{n+N}=\doubleu_n$ (when oriented)
or $\doubleu_{n+N}=-\doubleu_n$ (when anti-oriented) for any $n\in \Z$.

When $N$ is odd and the Kaleidocycle is oriented, the equation \eqref{condition:kappa_psi2} together with $\doubleu_0=\doubleu_N$
forms a linear system for $\doubleu_n \ (0\le n\le N)$ which is regular.
Therefore, we can find $\doubleu_n \ (0\le n\le N)$ uniquely as the solution to the system.
Then, the equation \eqref{def:deformation_curve} generates a deformation of
$\gamma$ which preserves the segment length and the torsion angle, while $\gamma$ remains closed.
That is, the turning-over motion of the Kaleidocycle is governed by the semi-discrete mKdV equation \eqref{eq:sine-Gordon_kappa2}
(see Fig. \ref{fig:centreline-curve}).
Note that by \eqref{eq:sine-Gordon_kappa2}, the total curvature angle $\sum_{i=0}^{N-1} \kappa_n(t)$ is also preserved.
\begin{figure}[ht]
    \center
  \includegraphics[width=6cm]{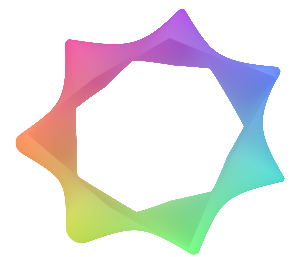}
  \includegraphics[width=6cm]{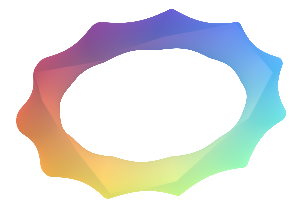}
  \caption{Surface drawn by the evolution of the center curves of Kaleidocycles with $N=7$ and $N=25$ respectively.}
  \label{fig:centreline-curve}
\end{figure}

When the Kaleidocycle is anti-oriented,
the equation \eqref{condition:kappa_psi1} together with $\doubleu_0=-\doubleu_N$
forms a linear system for $\doubleu_n \ (0\le n\le N)$ which is regular for any $N$.
Similarly to the above, in this case the turning-over motion of the Kaleidocycle is governed by the semi-discrete sine-Gordon
 equation \eqref{eq:sine-Gordon_kappa}.

Note that if an $N$-Kaleidocycle with an odd $N$ is anti-oriented $b_0=-b_N$,
we can define an oriented Kaleidocycle by
taking its ``mirrored image'' $b_i \mapsto (-1)^i b_i$
which conforms to the definition \ref{dfn:kaleidocycle}.
Thus, for an odd Kaleidocycle, both the semi-discrete mKdV equation and the semi-discrete sine-Gordon
 equation generate the turning-over motion.

\section{Extreme Kaleidocycles}\label{sec:final}
We defined Kaleidocycles in Def. \ref{dfn:kaleidocycle} and saw the torsion angle cannot be chosen
arbitrarily.  A natural question is for what torsion angle $\nu$ there exists an $N$-Kaleidocycle
for each $N$.  It seems there are no Kaleidocycles with $\nu \in (0,\pi)$ for $N\le 5$.  For $6\le N \le 50$, we
conducted numerical experiments with \cite{code} and found that there exists $c^*_N\in [0,1]$ which
satisfy the following.  Recall that $\pi_c: \mathcal{M}_N\to \R$ is the projection of the
configuration space $\mathcal{M}_N$ onto the $c$-axis, where $c=\cos\nu$.
\begin{enumerate}
\item When $N$ is odd,
$\pi_c(\mathcal{M}_N^+)=[-c^*_N,1]$ and $\pi_c(\mathcal{M}_N^-)=[-1,c^*_N]$.
\item When $N$ is even, $\pi_c(\mathcal{M}_N^+)=[-1,1]$ and $\pi_c(\mathcal{M}_N^-)=[-c^*_N,c^*_N]$.
\end{enumerate}
Moreover, $N \arccos{(c^*_N)}$ converges monotonously to a constant, where $\arccos$ takes the principal value in $[0,\pi]$.
Interestingly, at the boundary values $c=\pm c^*_N$,
the fibre of $\pi_c$ seems to be exactly one-dimensional for any $N\ge 6$.
This means, they are exactly the one-dimensional orbits defined in \S \ref{sec:motion_kaleidocycle}.

We summarise our numerical findings.
\begin{conjecture}
Let $N\ge 6$.
We have the following:
\begin{enumerate}
\item The space $\pi_c^{-1}(c^*_N)\cap \mathcal{M}_N^-$ is a circle.  Moreover, the involution
defined by $b_n \mapsto (-1)^n b_n$ induces isomorphisms $\pi_c^{-1}(-c^*_N)\cap \mathcal{M}_N^+
\simeq \pi_c^{-1}(c^*_N)\cap \mathcal{M}_N^-$ when $N$ is odd and $\pi_c^{-1}(-c^*_N)\cap
\mathcal{M}_N^- \simeq \pi_c^{-1}(c^*_N)\cap \mathcal{M}_N^-$ when $N$ is even.
\item The orbit of any element $\gamma\in \pi_c^{-1}(c^*_N)\cap \mathcal{M}_N^-$ of the flow generated 
by the semi-discrete sine-Gordon equation described in  \S \ref{sec:motion_kaleidocycle}
coincides with $\pi_c^{-1}(c^*_N)\cap \mathcal{M}_N^-\simeq S^1$.
\item When $N$ is odd, the orbit of any element $\gamma\in \pi_c^{-1}(-c^*_N)\cap \mathcal{M}_N^+$
generated by the semi-discrete mKdV equation described in \S \ref{sec:motion_kaleidocycle} coincides with
$\pi_c^{-1}(-c^*_N)\cap \mathcal{M}_N^+\simeq S^1$.
Moreover, on $\pi_c^{-1}(-c^*_N)\cap \mathcal{M}_N^+$ we have $\sum_{n=0}^{N-1}\kappa_n=0$ and we can also define its deformation by
the semi-discrete sine-Gordon equation if we define $\doubleu$ by \eqref{condition:kappa_psi1} and
$\sum_{n=0}^{N-1}\dot{\kappa}_n=2\alpha\sum_{n=0}^{N-1}\sin(\doubleu_n)=0$.  The orbit coincides
with $\pi_c^{-1}(-c^*_N)\cap \mathcal{M}_N^+$ as well.  That is, for an oriented Kaleidocycle with
$\nu=\arccos(-c^*_N)$, we can define two motions one by the semi-discrete sine-Gordon equation
\eqref{condition:kappa_psi1}, the other by the semi-discrete mKdV equation
\eqref{condition:kappa_psi2}, and they coincide up to rigid transformations.
\item Any strip $(\gamma^{b,\varepsilon},b)$ corresponding to $b\in \pi_c^{-1}(c^*_N)\cap
\mathcal{M}_N^-$ is a 3-half twisted M\"obius strip (see \S \ref{sec:topology}).
There are no Kaleidocycles with one or two half twisting.
\item When $N$ tends to infinity, $N \arccos{c^*_N}$ converges to a constant.
There exists a unique limit curve up to congruence
for any sequence $\gamma_N\in \pi_c^{-1}(c^*_N)\cap \mathcal{M}_N^-$, and it has a constant torsion up to sign.
\end{enumerate}
\end{conjecture}
We call those Kaleidocycles having the extremal torsion angle \emph{extreme Kaleidocycles}.

\begin{remark}
The extreme Kaleidocycles were discovered by the first named author and his collaborators
\cite{kaji,patent}.  In particular, when it is anti-oriented, it is called the \emph{M\"obius
Kaleidocycle} because they are a discrete version of the M\"obius strip with a $3\pi$-twist.
Coincidentally, M\"obius is the first one to give the dimension counting formula for generic
linkages \cite{mobius} (although it is often attributed to Maxwell), and our M\"obius Kaleidocycles
are exceptions to his formula.
\end{remark}

We end this chapter with a list of interesting properties, questions and some supplementary materials
of Kaleidocycles for future research.

\subsection{Kinematic energy}
Curves with adapted frames serve as a model of elastic rods
and are studied, for example, in Langer and Singer \cite{Langer-Singer} in a continuous setting,
and in \cite{DiscreteElasticRods} in a discrete setting.
Serial and closed hinged networks are discrete curves with specific frames
as we saw in \S \ref{circle-system}.
From this viewpoint, we consider some energy functionals defined for discrete curves with frames
and investigate how they behave on the configuration space $\mathcal{M}_N$ of Kaleidocycles.

Let $\gamma$ be a constant speed discrete closed curve of length $N$.
The {\em elastic energy} $\mathcal{E}_e$ and the {\em twisting energy} $\mathcal{E}_t$ are defined respectively by
\[
\mathcal{E}_e(\gamma) = \sum_{n=0}^{N-1} \kappa_n^2, \qquad 
\mathcal{E}_t(\gamma) = \sum_{n=0}^{N-1}\nu_n^2. 
\]
By the definition of Kaleidocycle, $\mathcal{E}_t$ takes a constant value when a Kaleidocycle
undergoes any motion.

Interestingly, a numerical simulation by \cite{code} suggests that on $\pi_c^{-1}(c^*_N) \cap
\mathcal{M}_N^-$ (and also on $\pi_c^{-1}(-c^*_N) \cap \mathcal{M}_N^+$ for an odd $N$ and 
on $\pi_c^{-1}(-c^*_N) \cap \mathcal{M}_N^-$ for an even $N$)
for a fixed $N$, $\mathcal{E}_e$ takes an almost constant value.  The summands of $\mathcal{E}_e$ are
locally determined and vary depending on the states, however, the total is almost stable so that
only small force should be applied to rotate the Kaleidocycle.  It is also noted that the sum
$\mathcal{E}_e+\mathcal{E}_t$ is a discrete version of the elastic energy of the Kirchoff rod
defined by the strip, and it also takes almost constant values.

Similarly, 
we introduce the following three more energy functionals, which are observed to take almost constant values 
on $\pi_c^{-1}(-c^*_N) \cap \mathcal{M}_N^+$.
The {\em dipole energy} is defined to be
\[
\mathcal{E}_d(\gamma) := 2\left( \sum_{i<j} \dfrac{\langle b_i, b_j\rangle}{|\gamma_i-\gamma_j|^3}-3\dfrac{\langle b_i,\gamma_i-\gamma_j\rangle
\langle b_j,\gamma_i-\gamma_j\rangle}{|\gamma_i-\gamma_j|^5}\right).
\]
The {\em Coulomb energy} with an exponent $\alpha>0$ is defined to be
\[
\mathcal{E}_c(\gamma) := 2\sum_{i<j} \dfrac{1}{|\gamma_i-\gamma_j|^\alpha}.
\]
The \emph{averaged hinge magnitude} is defined to be
\[
\mathcal{E}_a(\gamma) := \dfrac{1}{N}\left| \sum_{n=0}^{N-1} b_n \right|.
\]
However, we have no rigorous statements about them.  It may be the case that one
needs some other discretisation of the continuous counterparts of these energies to show their
behaviour theoretically.  It is also interesting to characterise or generalise extreme Kaleidocycles
in terms of variational calculus on the space of discrete closed curves.

\subsection{Topological invariants}\label{sec:topology}
As noted in \cite{Langer-Singer}, for a curve
to be closed, topological constraints come into the story.
This quantises some continuous quantity and makes it an isotopy invariant.

Let $\gamma$ be a constant speed discrete closed curve of length $N$.
First, interpolate $\gamma_n$ and $b_n$ for  $(0\le n < 2N)$ linearly
to obtain a continuous vector field $\bar{b}$ defined on the polygonal 
curve $\bar{\gamma}$, which goes around the polygon twice.
We define the \emph{twisting number} $\Tt$ of $\gamma$ as the linking number 
between twice the centre curve $\bar{\gamma}$ and the boundary curve $\bar{\gamma}+\epsilon\bar{b}$,
where $\epsilon>0$ is small enough.
\begin{figure}[ht]
  \center
\includegraphics[width=4.5cm]{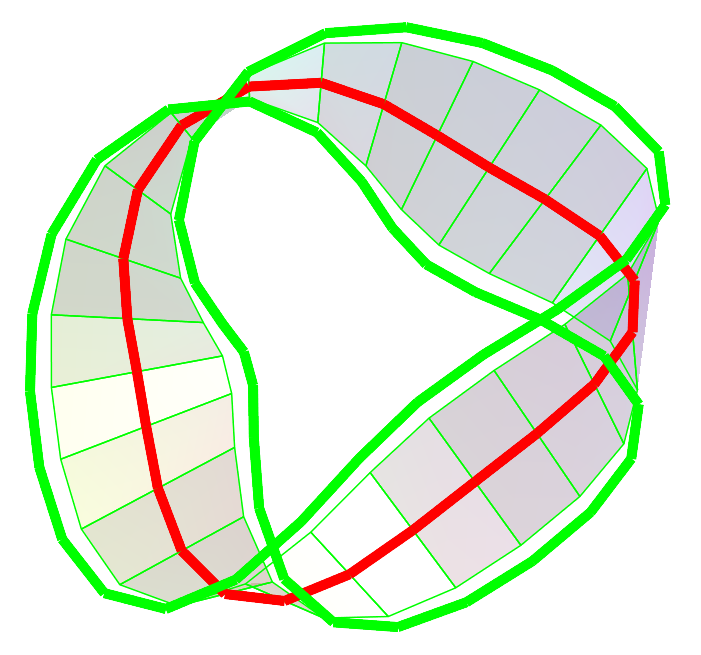}
\caption{Twisting number as the linking number between centre and boundary curves.}\label{fig:twisting}
\end{figure}
Intuitively, it is the number of half-twists of the strip defined by $\gamma$ and $b$.
The C\v alug\v areanu-White formula 
relates this topological invariant to the sum of two conformal invariants
and provides a direct discretisation without the need of interpolation (cf. \cite{Klenin-Langowski_2000}):
\begin{equation}\label{total-twist}
\Tt = 2(\Tw + \Wr),
\end{equation}
where $\Wr$ is the \emph{writhe} of the polygonal curve $\gamma$ which can be computed as a double summation
\cite[Eq. (13)]{Klenin-Langowski_2000}
and 
\[
\Tw=\dfrac{1}{2\pi} \sum_{n=0}^{N-1} \nu
\]
is the {\em total twist}. 
The twisting number $\Tt$ takes values in the integers, enforcing topological constraints
to the curve.

Recall by definition that anti-oriented extreme Kaleidocycles are 
discrete closed space curves of constant speed and constant torsion
which have the minimum odd twisting number.
Our numerical experiments suggest that the minimum is not one but three.

Let $\gamma$ be a discrete closed space curve of constant speed and constant torsion
corresponding to a Kaleidocycle.
Under any motion of the Kaleidocycle, $\Tw$ stays constant by definition. 
By \eqref{total-twist} the corresponding deformation
of the curve preserves the writhe as well.
This can equivalently be phrased in terms of the \emph{Gauss map} 
$G(\gamma): n\mapsto T_n \ (0\le n\le N-1)$.
The Gauss-Bonnet theorem tells us 
that $A + 2\pi \Tw = 0 \mod \pi$, where $A$ is the area on the sphere enclosed by $G(\gamma)$.
By \eqref{total-twist} we have $\Wr = A/2\pi \mod 1/2$.
Thus, the deformation of the closed discrete space curve considered in \S \ref{sec:motion_kaleidocycle}
induces one of the closed discrete spherical curves which preserves the enclosed area $A$.

\medskip 
Kaleidocycles can be folded from a piece of paper.  We include a development plan for the extreme
Kaleidocycle with $N=8$ so that the readers can personally make and investigate its motion.
\begin{figure}[ht]
\center
\includegraphics[height=12cm]{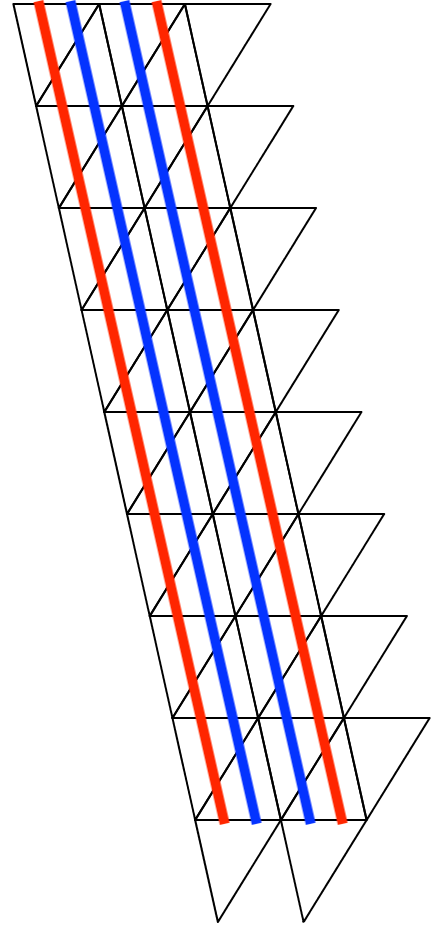}
\caption{Development plan of an extreme Kaleidocycle with eight hinges.
Black horizontal lines indicate valley folds and black slanted lines indicate mountain folds. }
\end{figure}


\begin{acknowledgments}
The first named author is partially supported by JST, PRESTO Grant Number JPMJPR16E3, Japan.
The second named author is partially supported by JSPS KAKENHI Grant Numbers JP16H03941, JP16K13763.
The last named author acknowledges the support from the ``Leading Program in Mathematics for Key Technologies" of Kyushu University. 
\end{acknowledgments}

\strut\hfill



\begin{thebibliography}{99}
  \bibitem{bates2013}
Bates L M and Melko O M,
{\it On curves of constant torsion I},
Journal of Geometry, \textbf{104}(2), 213--227, 2013.

\bibitem{DiscreteElasticRods}
Bergou M, Wardetzky M, Robinson S, Audoly B and Grinspun E,
{\it Discrete elastic rods}, ACM Trans. Graph., \textbf{27}(3), Article 63, 2008.

\bibitem{Boiti:dsG} Boiti M, Pempinelli F and Prinari B,
{\it Integrable discretization of the sine-Gordon equation},
Inverse Problems {\bf 18}(5), 1309--1324, 2002.

\bibitem{Byrnes}
Byrnes R,
{\it Metamorphs: Transforming Mathematical Surprises},
Tarquin Pubns, 1999.

\bibitem{Calini-Ivey1996}
Calini A M, Ivey T A,
{\it B\"acklund transformations and knots of constant torsion},
 J. Knot Theory Ram., \textbf{7}, 719--746, 1998.

 \bibitem{Calini-Ivey1999}
Calini A M an Ivey T A,
{\it Topology and sine-Gordon evolution of constant torsion curves},
Phys. Lett. A, \textbf{254}(3--4), 170--178, 1999.

\bibitem{Chen2011}
You Z and Chen Y,
{\it Motion Structures: Deployable Structural Assemblies of Mechanisms},
Taylor \& Francis, 2011.

\bibitem{darboux1917}
Darboux G,
{\it Le\c{c}cons sur la Th\'eorie G\'en\'erale des Surfaces},
Gauthier-Villars, 1917.

\bibitem{DH}
Denavit J and Hartenberg R S,
{\it A kinematic notation for lower-pair mechanisms based on matrices},
Trans ASME J. Appl. Mech. \textbf{23}, 215--221, 1955.

\bibitem{Doliwa-Santini:PLA}
Doliwa A and Santini P M, {\it An elementary geometric characterization of the integrable motions of a curve},
Phys. Lett. A, {\bf 185}, 373--384, 1994.

\bibitem{Doliwa-Santini:JMP} Doliwa A and Santini P M,
{\it Integrable dynamics of a discrete curve and the Ablowitz-Ladik hierarchy}, J. Math. Phys. {\bf 36}, 1259--1273, 1995.

\bibitem{Faber2008}
Farber M,
{\it Invitation to Topological Robotics},
 European Mathematical Society, Zurich, 2008.


\bibitem{Foweler-Guest2005}
Fowler P W and Guest S D,
{\it A symmetry analysis of mechanisms in rotating rings of tetrahedra},
Proc. R. Soc. A {\bf 461}, 1829--1846, 2005.

\bibitem{Freuder}
Freuder E C,
{\it Synthesizing constraint expressions},
Commun. ACM \textbf{21}(11), 958--966, 1978.

\bibitem{Hasimoto}
Hasimoto H, {\it A soliton on a vortex filament}, J. Fluid. Mech. {\bf 51}, 477--485, 1972.

\bibitem{Hisakado-Wadati} Hisakado M and Wadati M,
{\it Moving discrete curve and geometric phase}, Phys. Lett. A, {\bf 214}, 252--258, 1996.

\bibitem{Hoffmann:MI}
Hoffmann T, {\it Discrete Differential Geometry of Curves and Surfaces}, 
MI Lecture Notes vol. 18,  Kyushu University, 2009.

\bibitem{inoguchi2014discrete}
  Inoguchi J, Kajiwara K, Matsuura N and Ohta Y,
  \textit{Discrete mKdV and discrete sine-Gordon flows on discrete space curves},
  J. Phys. A: Math. Theor. \textbf{47}, 2014.

 \bibitem{Ivey}
Ivey T A,
{\it Minimal Curves of Constant Torsion},
Proc. AMS,
\textbf{128}(7), 2095--2103, 2000.

\bibitem{jordan}
Jord\'an T, Kir\'aly C and Tanigawa S,
{\it Generic global rigidity of body-hinge frameworks},
J. Comb. Theory B, \textbf{117},
59--76, 2016.

\bibitem{kaji}
Kaji S,
{\it A closed linkage mechanism having the shape of a discrete M\"obius strip},
the Japan Society for Precision Engineering Spring Meeting Symposium Extended Abstracts,
62--65, 2018. The original is in Japanese but an English translation is available at arXiv:1909.02885.

\bibitem{code}
Kaji S,
{\it Geometry of the moduli space of a closed linkage: a Maple code},
available at \url{https://github.com/shizuo-kaji/Kaleidocycle}

\bibitem{patent}
Kaji S, Sch\"onke S, Grunwald M and Fried E,
{\it M\"obius Kaleidocycle}, patent filed, JP2018-033395, 2018.

\bibitem{Kapovich-Millson2002}
Kapovich M and Millson J,
{\it Universality theorems for configuration spaces of planar linkages},
Topology \textbf{41}(6), 1051--1107, 2002.

\bibitem{katoh}
Katoh N and Tanigawa S,
{\it A proof of the molecular conjecture},
Discrete Comput Geom., \textbf{45}, 647--700, 2011.

\bibitem{Klenin-Langowski_2000}
Klenin K and Langowski J,
{\it Computation of writhe in modeling of supercoiled DNA},
Biopolymers, \textbf{54}, 307--317, 2000.

\bibitem{Lamb}
Lamb G L Jr., {\it Solitons and the motion of helical curves}, Phys. Rev. Lett. {\bf 37}, 235--237, 1976.

\bibitem{Langer-Perline} Langer J and Perline R, {\it Curve motion inducing modified Korteweg-de Vries systems},
Phys. Lett. A, {\bf 239}, 36--40, 1998.

\bibitem{Langer-Singer}
Langer J and Singer D A,
{\it Lagrangian aspects of the Kirchhoff elastic rod},
SIAM Review \textbf{38}(4), 605--618, 1996.

\bibitem{lavalle}
LaValle S M,
{\it Planning algorithms},
Cambridge University Press, 2006.

\bibitem{Magalhaes}
Magalh\~aes M L S and Pollicott M,
{\it Geometry and dynamics of planar linkages},
Comm. Math. Phys.,
\textbf{317}(3), 615--634, 2013.

\bibitem{mobius}
M\"obius A F,
{\it Lehrbuch der Statik}, \textbf{2}, Leipzig, 1837.


\bibitem{Moses2013}
Moses M S and Ackerman M K and Chirikjian G S,
{\it ORIGAMI ROTORS: Imparting continuous rotation to a moving platform using compliant flexure hinges},
Proc. IDETC/CIE 2013.

\bibitem{Muller2015}
M\"uller A,
{\it Representation of the kinematic topology of mechanisms for kinematic analysis},
Mech. Sci., \textbf{6}, 1--10, 2015.

\bibitem{Muller2016}
M\"uller A,
{\it Local kinematic analysis of closed-loop linkages -- mobility, singularities, and shakiness},
J. Mechanisms Robotics \textbf{8}(4), 041013, 2016.

\bibitem{Nakayama:JPSJ2007} Nakayama K,
{\it Elementary vortex filament model of the discrete nonlinear Schr\"odinger equation},
J. Phys. Soc. Jpn. {\bf 76}, 074003, 2007.

\bibitem{Nakayama_Segur_Wadati:PRL}
Nakayama K, Segur H and Wadati M, {\it Integrability and the motions of curves},
Phys. Rev. Lett. {\bf 69}, 2603--2606, 1992.

\bibitem{Nishinari} Nishinari K,
{\it A discrete model of an extensible string in three-dimensional space},
J. Appl. Mech. {\bf 66}, 695--701, 1999.

\bibitem{Orfanidis1} Orfanidis S J,  {\it Discrete sine-Gordon equations},
Phys. Rev. D, {\bf 18}(10), 3822--3827, 1978.

\bibitem{Orfanidis2} Orfanidis S J,  {\it Sine-Gordon equation and nonlinear $\sigma$ model on a lattice},
Phys. Rev. D, {\bf 18}(10), 3828--3832, 1978.

\bibitem{Rogers-Schief:book}
Rogers C and Schief W K,
{\it B\"acklund and Darboux Transformations: Geometry and Modern Applications in Soliton Theory},
Cambridge University Press, Cambridge, 2002.

\bibitem{Sato-Tanaka}
Sato K and Tanaka R,
{\it Solitons in one-dimensional mechanical linkage}
Phys. Rev. E, \textbf{98}, 2018.

\bibitem{Schattschneider1987}
Schattschneider D and Walker W M,
{\it M. C. Escher Kaleidocycles}, Pomegranate Communications: Rohnert Park, CA, 1987.
(TASCHEN; Reprint edition, 2015).

\bibitem{Bertini}
Sommese A J, Hauenstein J D, Bates D J and Wampler C W,
{\it Numerically Solving Polynomial Systems with Bertini},
Software, Environments, and Tools, Vol. 25, SIAM, Philadelphia, PA, 2013.


\bibitem{Weiner1974}
Weiner L J,
{\it Closed curves of constant torsion},
Arch. Math. (Basel) \textbf{25}, 313--317, 1974.

\bibitem{Weiner1977}
Weiner L J,
{\it Closed curves of constant torsion II},
Proc. AMS, \textbf{67}(2), 1977.


\end{thebibliography}
\end{document}